%% file: aamas_2023_action_deception.tex
\documentclass[sigconf]{aamas} 

\usepackage{acronym}
\usepackage{amsthm}
\usepackage{mathrsfs}
\usepackage{balance} 
\input{defs}
\acrodef{bsr}[BSR]{behaviorally subjectively rationalizable}

\definecolor{darkgreen}{rgb}{0,0.5,0}

\newtheorem{problem}{Problem}
\newtheorem{example}{Example}
\newtheorem{remark}{Remark}
\newtheorem*{remark*}{Remark}

\newtheorem{assumption}{Assumption}

\setcopyright{ifaamas}
\acmConference[AAMAS '23]{Proc.\@ of the 22nd International Conference
on Autonomous Agents and Multiagent Systems (AAMAS 2023)}{May 29 -- June 2, 2023}
{London, United Kingdom}{A.~Ricci, W.~Yeoh, N.~Agmon, B.~An (eds.)}
\copyrightyear{2023}
\acmYear{2023}
\acmDOI{}
\acmPrice{}
\acmISBN{}

\acmSubmissionID{1076}

\title[Planning with Action Deception]{Quantitative Planning with Action Deception in Concurrent Stochastic Games}

\author{Chongyang Shi}
\affiliation{
  \institution{University of Florida}
  \city{Gainesville}
  \country{United States}}
\email{c.shi@ufl.edu}

\author{Shuo Han}
\affiliation{
  \institution{University of Illinois Chicago}
  \city{Chicago}
  \country{United States}}
\email{hanshuo@uic.edu}

\author{Jie Fu}
\affiliation{
  \institution{University of Florida}
  \city{Gainesville}
  \country{United States}}
\email{fujie@ufl.edu}

\begin{abstract} 
We study a class of two-player competitive concurrent stochastic games on graphs with reachability objectives. Specifically, player 1 aims to reach a subset $F_1$ of game states, and player 2 aims to reach a subset $F_2$ of game states where $F_2\cap F_1=\emptyset$. Both players aim to satisfy their reachability objectives before their opponent does. Yet, the information players have about the game dynamics is asymmetric: P1 has a (set of) hidden actions unknown to P2 at the beginning of their interaction. In this setup, we investigate P1's strategic planning of action deception that decides when to deviate from the Nash equilibrium in P2's game model and employ a hidden action, so that P1 can maximize the value of action deception, which is the additional payoff compared to P1's payoff in the game where P2 has complete information. Anticipating that P2 may detect his misperception about the game and adapt his strategy during interaction in unpredictable ways, we construct a planning problem for P1 to augment the game model with an incomplete model about the theory of mind of the opponent P2. While planning in the augmented game, P1 can effectively influence P2's perception so as to entice P2 to take actions that benefit P1. We prove that the proposed deceptive planning algorithm maximizes a lower bound on the value of action deception and demonstrate the effectiveness of our deceptive planning algorithm using a robot motion planning problem inspired by soccer games.
\end{abstract}


\newcommand{\BibTeX}{\rm B\kern-.05em{\sc i\kern-.025em b}\kern-.08em\TeX}

\begin{document}


\pagestyle{fancy}
\fancyhead{}


\maketitle 


\section{Introduction}
Asymmetrical information between players is commonly encountered in conflict analysis and security applications \cite{carrollGameTheoreticInvestigation2009,hespanha2000deception,schlenker_deceiving_2018,nguyen_deception_2019}. For adversarial interactions, a player can leverage the asymmetric information, or his/her opponent's disinformation to gain additional benefits towards achieving his/her objective. Such strategic reasoning and planning are termed deception \cite{estornell_deception_2020}.

In general, deception techniques can be categorized into two classes: One is intention deception where the mark has misinformation or disinformation about the intention of the deceiver.  In game theory, intention deception is also known as payoff misperception where players have different perceptions about the payoff matrix in the game.
Many existing literature focus on this class of deception \cite{gharesifardLearningEquilibriaMisperceptions2011,karabagDeceptionSupervisoryControl2021a,kulkarniSignalingFriendsHeadFaking2020} with applications to cyber defense with deception \cite{thakoorCyberCamouflageGames2019a,thakoorExploitingBoundedRationality2020,duLearningPlayAdaptive2022} and motion planning  \cite{liDynamicHypergamesSynthesis2020}.  Another class is capability deception, where the mark has incomplete knowledge about the deceiver's action or perception capabilities. 

For modeling the interactions with asymmetric information, Bayesian games \cite{harsanyiGamesIncompleteInformation1967a} and hypergames \cite{bennettHypergameTheoryMethodology1986} have been used \cite{al-shaerDynamicBayesianGames2019,huangDynamicGamesApproach2020,bakkerHypergamesCyberPhysicalSecurity2019,gharesifardLearningEquilibriaMisperceptions2011}. In Bayesian game, the solution approach is to transform a game with incomplete information to a game with imperfect information by capturing players incomplete information as a type variable, which is not observable to
other players. However, in many deception settings, the type space, \ie, the hypothesis space of the opponent's intention or capabilities are not public knowledge. Hypergame on the other hand, construct a hierarchy of perceptual games that captures the hierarchical information --- what player $i$ knows about the game known to the other player. 

Since a hypergame model allows us to capture the unawareness of players, this work extends the hypergame model to analyze action deception in a class of two-player concurrent, stochastic games where each player has a reachability objective, represented by a set of goal states to be reached. Specifically, as the game starts, player 1 (P1) has a (set of) private actions which are hidden from P2. Additionally, P1 has complete information about P2's actions. Hence, P1 may deploy action deception --- deciding when is the best time to reveal a hidden action to capitalize on the gain from P2's suboptimal decisions caused by P2's incomplete information. In literature, two-player reachability games have been extensively studied for the case in which both the players have symmetric and complete information \cite{mcnaughton1993infinite,zielonka1998infinite, deAlfaro2000,chatterjee2012survey}. Quantitative solutions of reachability games can be formulated using Markov games \cite{filar2012competitive,abeOffPolicyExploitabilityEvaluationTwoPlayer2021}
 
Action deception in reachability games has been studied in  \cite{kulkarni_synthesis_2020}. The authors show that given two-sided perfect observations, the deceiver has a strategy to reach its goal    with probability one  by strategically revealing the private actions, when starting from a state that this objective cannot be achieved with probability one, had P2 known P1's private actions. However,  their solution is qualitative for turn-based games, whereas we investigate the gain of action deception for quantitative planning in concurrent games: When P1 cannot achieve the objective with probability one, how can P1 use action deception to improve his chance of achieving his objective? 

A key observation is that in quantitative planning, P2 may know there is a mismatch between the game she knows and the true game, by detecting a deviation of the gameplay from the predicted distribution resulting from a Nash equilibrium. Therefore, we formulate action deception to determine a switching time of two strategies:
 At the start of the interaction, P1 can select from two strategies: One is his best response $\pi_1^2$ in the game known to P2, called P2's perceptual game, and another is his best response $\pi_1$ in the true game, where a hidden action can be used. P1 is to determine the optimal timing to switch from $\pi_1^2$ to $\pi_1$ to maximize the \emph{value of action deception}, measured by the difference between P1’s payoff gained by deceiving P2 and P1’s payoff from the equilibrium in the true game where P2 knows about P1’s action set.

We develop a deceptive planning algorithm that  incorporates a theory of mind (ToM) of P2 with two components:
1) 
P2’s change detection mechanism: Given the switching time $t$, what is the delay $k_1$ that P2 may have to detect that P1 has deviated from her perceptual game?
2) 	P2’s reaction to deviation: What is P2’s strategy in reaction to P1’s action deception? We employ the solution concept of \emph{ subjective rationalizable strategies} \cite{sasakiSubjectiveRationalizabilityHypergames2014} to model P2's response, that is, P2 always behaves rationally in her perception of the game. However, P1's theory of mind for P2 is incomplete as P1 has no prediction of P2's response when P2 detects the mismatch but has not yet learned P1's private actions. Therefore, the value of action deception is defined to be the optimal gain against all completions of P1's incomplete ToM for P2.
By augmenting P1's planning state space with the additional state variables to track the ToM of P2, we showed that the optimal solution in the constructed planning problem provides a lower bound on the value of action deception using a switching strategy.


\section{Preliminaries and Problem Formulation}

\paragraph{Notations} Let $\reals$ denote the set of real numbers and $\reals^n$ the set of real $n$-vectors. Given a finite set $Z$, the set of probability distributions over $Z$ is represented as $\dist{Z}$. 


\subsection{Concurrent Stochastic Games with Reachability Objectives}
We start by introducing a standard model of two-player stochastic games played on a graph with perfect observations. It consists of two components: A game graph describing the players' interacting dynamics, and a pair of players' intentions/objectives expressed as reachability properties, that is, each player has a set of goal states to be reached. We refer to player 1 as P1 (pronoun he/him/his) and player 2 as P2 (pronoun she/her/hers).

\begin{definition}[Concurrent stochastic games on graphs with reachability objectives]
A two-player, concurrent, stochastic game on a graph is a tuple
\[
G= (S, A, P, s_0, \gamma, F_1,F_2),
\]
with the following components:
\begin{itemize}
 \item $S$ is a finite set of states.
\item $A = A_1 \times A_2$ is a finite set of actions, where $A_1$ is the set of actions that P1 can perform, and $A_2$ is the set of actions that P2 can perform.
\item $P \colon S \times A \to \dist S$ is a probabilistic transition function. At every state $s \in S$, P1 chooses an action $a  \in A_1$, and P2 chooses an action $b \in A_2$ simultaneously. Then, a successor state $s'$ is determined by the probability distribution $P(\cdot \mid s,(a,b))$.
\item $s_0$ is an initial state.
\item $\gamma \in (0,1]$ is a discounting factor;
  \item  $F_1\subseteq S$ is P1's target states, $F_2\subseteq S \setminus F_1$ is referred to as P2's target states.  A reachability objective with the target set $F_i$ means that the player aims to reach a state in $F_i$. All states in $F_1\cup F_2$ are sink/absorbing states, regardless of players' actions.
\end{itemize}
\end{definition} 
In the following, we refer to the game as a concurrent reachability game.
A \emph{play} in the game is constructed as follows: The players start in the initial game state $s_0$, simultaneously select a pair of actions $(a,b)\in A$, and with some probability, move to a next state $s_1$, and repeat. The game ends when one of the players satisfies his/her objective. Thus, a play $\rho $  is a sequence of states and actions $s_0 (a_0,b_0) s_1 (a_1,b_1) \ldots $ such that $P(s_{i+1} \mid s_i, (a_i,b_i))>0$ for any $ i\ge 0$. A prefix of a play is a finite initial segment of the state-action sequence.   The set of all possible plays in the game is denoted by $\plays$. The set of prefixes of plays is denoted by $\prefplays$.

The reachability objective with the target set $F$ can be expressed by the temporal logic formula $\Eventually F$ read ``eventually $F$''. The symbol $\Eventually$ is a temporal operator for the eventuality. \footnote{Since we consider only reachability objectives and use the formula to simplify some notations, we omit the preliminaries for temporal logic, which can be found in \cite{mannaTemporalLogicReactive1992}.} 
A play $\rho=s_0 (a_0,b_0) s_1 (a_1,b_1) \ldots $ is said to satisfy the formula, denoted by $\rho \models \Eventually F$, if there exists $i\ge 0$, $s_i\in F$. 
We denote by $\llbracket \Eventually F
\rrbracket = \{\rho \in \plays \mid \rho \models \Eventually F \}$ the set of plays that satisfy the reachability objective defined by a target set $F$. Let $\rho[i]$ be the $i$-th state in the play $\rho$. For any state $s\in S$, the set of plays starting from $s$ and satisfying the formula, that is, $\{\rho \in \plays \mid \rho \models \Eventually F, \rho[0]=s\}$, can be shown to be measurable.  

A (mixed) strategy $\pi_i \colon \prefplays \to \dist {A_i}$, for player $i \in\{1,2\}$, is a function that assigns a probability
distribution over all actions given a prefix of a play. Let $\Pi_i$ denote the (mixed) strategy space of player $i$. A strategy profile $\langle \pi_1, \pi_2 \rangle$ is a pair of strategies, one for each
player. A strategy profile $\langle \pi_1, \pi_2 \rangle$ induces a probability measure $\Pr^{\langle \pi_1,\pi_2 \rangle}$ over $\prefplays$.

We say that player $i$ \emph{almost surely}  wins the game if the player can ensure, no matter how the opponent plays, that a state in $F_i$ will be reached \emph{with probability one}. We formally define the almost-sure winning region and strategy as follows.
\begin{definition}[Almost-sure winning strategy/region \cite{deAlfaro2000}]
A strategy $\pi_1$ is almost-sure winning for P1 starting from state $s\in S$ if and only if $\Pr_s^{\langle \pi_1,\pi_2\rangle} (\llbracket \Eventually F_1
\rrbracket ) =1$ for any $\pi_2\in \Pi_2$, where $\Pr_s^{\langle \pi_1,\pi_2\rangle}$ is the probability measure over paths starting from $s$ induced by the strategy profile $\langle\pi_1,\pi_2\rangle$.
The winning region of player $i$ is defined by $\asw_i = \{s\in S\mid \exists \pi_1\in \Pi_1,\forall \pi_2\in \Pi_2, \Pr_s^{{\langle \pi_1,\pi_2\rangle}} (\llbracket \Eventually F_i
\rrbracket ) =1\}$, which is the set of states 
starting from which, there exists an almost-sure winning strategy.
\end{definition}
The almost-sure winning region and strategy for P2 are defined analogously, with respect to P2's reachability objective $\Eventually F_2$. An algorithm for computing the almost-sure winning regions and strategies for concurrent stochastic games with reachability objectives can be found in \cite{deAlfaro2000}. Further, the game is memoryless determined. 

\begin{lemma}[\cite{deAlfaro2000}]
In a concurrent reachability game, for any $s\in \asw_1$, there exists a memoryless, almost-sure winning strategy for P1 starting from $s$.
\end{lemma}

When the game reaches a state in player $i$'s almost-sure winning region, player $i$ is ensured to eventually reach the target set $F_i$ by following his/her almost-sure winning strategy, whereas player $j$ has no chance of reaching the set $F_j$ given that $F_1 \cap F_2= \emptyset$. 

For any state $s\in S\setminus ( \asw_1 \cup \asw_2)$, both players have a  positive probability to reach their respective target sets. To compute a strategy for any state $s\in S\setminus  (\asw_1 \cup \asw_2)$, we introduce the following utility function and the concept of Nash equilibrium.

\begin{definition}[Utility functions]
The utility function for player $i$ is defined as 
$u_i\colon S \times \Pi_i\times \Pi_j \to \reals$ such that for $(i, j) \in \{(1, 2), (2,1)\}$, \[u_i(s, \pi_i,\pi_j)= 
\Expect^{\langle \pi_i,\pi_j \rangle}\left[ \sum_{t=1}^ \infty \gamma^t \cdot R(S_{t-1},(A_{t-1}, B_{t-1}), S_t)\mid S_0 =s
\right],\]
where $\gamma\in (0,1]$ is the discounting factor, and  $\{(S_i, A_i, B_i); i=0,1,\ldots\}$ taking value in $S\times A_1\times A_2$ is the stochastic process induced by the strategy profile $\langle \pi_i,\pi_j\rangle$ from the concurrent reachability game $G$, and 
$R: S\times A\times S\rightarrow \reals$ is the reward function defined as $R(s,(a,b),s') = 1$ if  $s'\in \asw_i$ and $s\notin \asw_i$ and  $R(s,(a,b),s') =0$ otherwise. 
\end{definition} 
In words, for any state in the positive winning region $s\in S\setminus (\asw_1\cup \asw_2)$, the utility of player $i$ at state $s$ measures the \emph{discounted} probability of reaching his/her almost-sure winning region $\asw_i$ from the state $s$.  
\begin{definition}[Nash equilibrium~\cite{filar2012competitive}]
A Nash equilibrium (NE) of a stochastic game  $G$ is a strategy profile $\langle \pi^\ast_1, \pi^\ast_2 \rangle$ with the property that for $(i, j) \in \{(1, 2), (2,1)\}$ we have 
\[
u_i(s, \pi_i^\ast, \pi_j^\ast) \ge u_i(s, \pi_i, \pi_j^\ast), \forall s \in S, \forall \pi_i \in \Pi_i.
\]
\end{definition}
The NE can be solved using the solutions of zero-sum Markov games \cite{filar2012competitive}.

\subsection{Problem Formulation of Action Deception}

We start by introducing asymmetric information in the game, which enables P1's deceptive planning.

\paragraph{Information Structure} The information owned by a player 
describes not only what the player observes during his/her interaction with the opponent, but also what the player knows about the components of the game. The following information structure is considered:
\begin{itemize}
    \item Both P1 and P2 have complete observations of states.
    \item P2 cannot observe P1's actions but P1 can observe P2's actions.
    \item P1's action set known to P2, denoted $A_1^2$, is a \emph{proper subset} of $A_1$, \ie, $A_1^2\subseteq A_1$.
    \item P1 knows both $A_1^2$ and $A_1$. 
\end{itemize}
\begin{remark}
The assumption that P2 cannot observe P1's actions can be relaxed, as we shall see in the planning algorithm, even if P2 may be able to observe P1's actions, there could still be an advantage for P1 to use action deception.
\end{remark}

Here is an informal problem statement.
\begin{problem}
Given the information structure between P1 and P2, how can P1  exploit P2's lack of information about P1's actions for strategic advantages?
\end{problem}

\subsection{An Illustrative Example: Soccer Game}
\label{sec:soccer}
We introduce a running example named soccer game to explain the above concepts. In this game, the field is a $3 \times 5$ grid. There are two players P1 and P2 in the game (A and B in Figures~\ref{fig:Action} and~\ref{fig:Target}). The ball (the star on the players) is possessed exclusively by one of the players. The two players move simultaneously. 
\begin{figure}[!h]
  \includegraphics[width=0.6\linewidth]{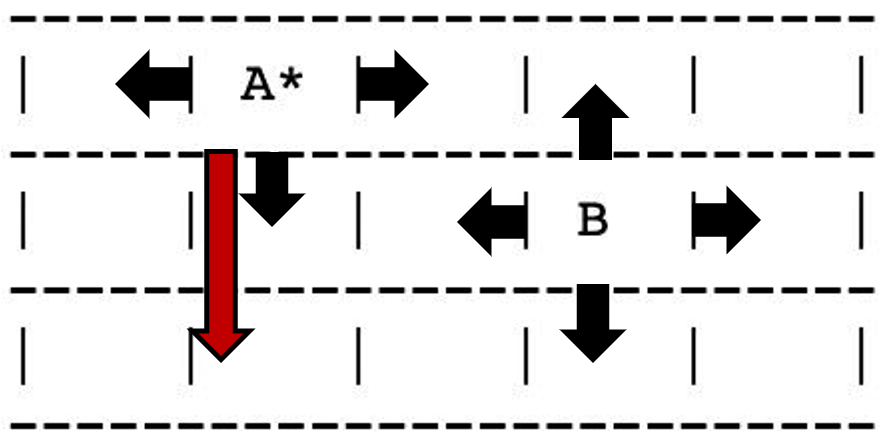}
  \caption{A soccer game between P1 (A) and P2 (B).}
  \label{fig:Action}
\end{figure}
\begin{figure}[!h]
  \includegraphics[width=0.6\linewidth]{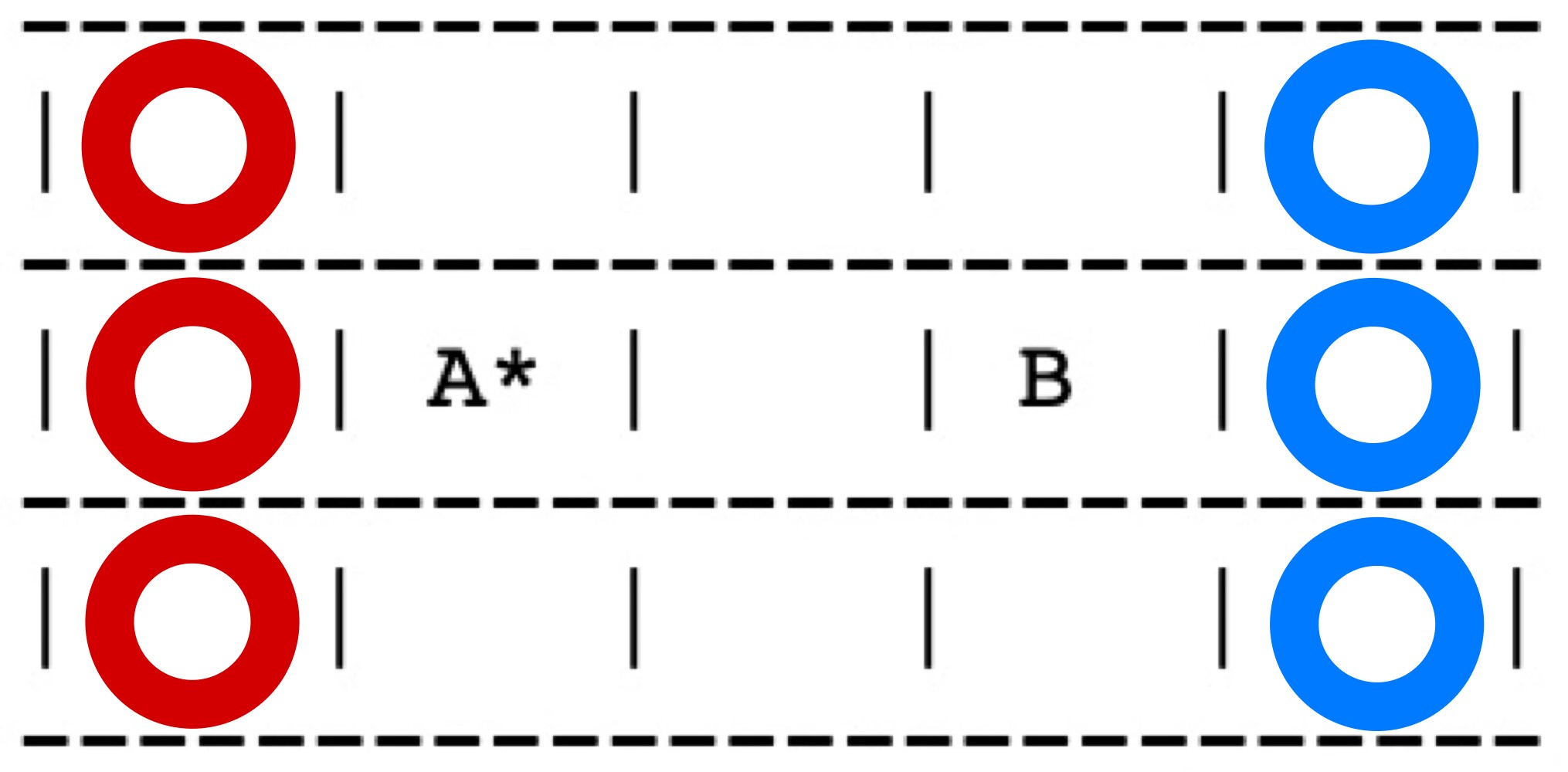}  

  \caption{The target states.}
  
  \label{fig:Target}
\end{figure}
The players can move up, down, left, and right (black arrows on Figure~\ref{fig:Action}), which are denoted by $a_U$, $a_D$, $a_L$, $a_R$, respectively. These actions are known to each player. So P1's action set is known to P2 is $A_1^2 = \{a_U,a_D,a_L,a_R\}$. Besides, P1 has a hidden action that he can move two cells down  (red arrow on Figure~\ref{fig:Action}). The hidden action is denoted by $a_H$. So the true action space of P1 is $A_1 = \{a_U,a_D,a_L,a_R,a_H\}$.

There are two notable rules of a soccer game.
\begin{itemize}
    \item[1.] The players cannot go out of bounds. If they select an action to do so, they will be forced to stay still. 
    \item[2.] When P1 and P2 move to the same cell, they each have a probability of $50$\% to get the ball.
\end{itemize}

The two players aim to bring the ball to their opponent's gate. That is, P1 needs to bring the ball to the blue circles, and P2 needs to bring the ball to the red circles. We denote the position of each player as a tuple $(i,j)$ where $i$ is the player's row and $j$ is the player's column. Let $p_1$ be the position of P1 and $p_2$ be the position of P2. Then we define a boolean variable $\ball = 0,1$. In this way, we can define a state $s$ of a soccer game as $s = (p_1, p_2, \ball)$. Therefore, P1's target set $F_1 = \{((i,4), p_2, 1)\}$ for any $i$ and $p_2$. P2's target set $F_2 = \{(p_1, (i,0), 0)\}$ for any $i$ and $p_1$.

At the beginning of the game, P2 does not know that P1 has the hidden action $a_H$. However, if P1 reveals his hidden action, P2 will update her knowledge and change her strategy. Thus, the question is how  P1 utilizes his hidden action to improve his chance of winning this game? 

\section{Planning with action deception}
In this section, we introduce our main algorithm for action deception planning.
\subsection{Hypergame Modeling and the Value of Action Deception}
We introduce a hypergame model to characterize the players' interaction given their respective information and higher-order information (that is, P1 knows about P2's incomplete information). First, it is observed that due to incomplete knowledge about P1's actions, P2's game graph is incomplete. This incomplete view is used to construct a \emph{perceptual game} for P2.

\begin{definition}[P2's perceptual game]
Given that P2 only knows a subset $A_1^2$ of P1's actions in the game graph $G$, P2's \emph{perceptual game} is defined by $G^2= (S, A_1^2 \times A_2, P^2, s_0, \gamma, F_1, F_2)$ where the transition function $P^2$ is obtained by eliminating all transitions enabled by P1's actions in $A_1\setminus A_1^2$ and any P2's actions in $A_2$. Formally, 
\begin{equation}
    \label{eq:P2trans}
    P^2(s,(a,b))= 
    \begin{cases}
    P(s,(a,b)) &\text{if $a\in A_1^2$, }\\
    \uparrow & \text{ otherwise.}
    \end{cases}
\end{equation}
where $\uparrow$ means the function is undefined for the given input.
\end{definition} 
In P2's perceptual game, for player $i \in \{1,2\}$, the best response strategy $\pi^2_i$ defined for $s\in S\setminus (\asw^2_1\cup \asw^2_2)$ together with the \ac{asw} strategy $\pi_i^{2,\asw}$ are \emph{subjective rationalizable}, because these are best response strategies to the opponent  in P2's perceptual game. Here $\asw^2_i$ is the almost-sure winning region of P2's perceptual game for player $i$ and $\asw_i$ is the almost-sure winning region of true game for player $i$.

To capture the asymmetric information, we extend the hypergame model  \cite{Bennett1977} to our game setup. 
\begin{definition}[Hypergame] 
Given the information structure considered herein, the interaction between P1 and P2 is captured by the hypergame 
\[
H^2 = (H^1  , G^2),
\]
where $H^1 = (G,G^2)$ is P1's perceptual game, which is a level-1 hypergame. The game $G^2$ is P2's perceptual game.
\end{definition}
 
In this level-2 hypergame, P1 knows both the true game $G$ and P2's perceptual game $G^2$. P2 knows only her perceptual game $G^2$. 




Following the notion of action deception, the deceiver hides his actions from the mark for some time and then deviates from the mark's perceptual game, by, for example, employing a strategy that uses the hidden action. We formalize the deceptive planning to determine when to deviate. For clarity, the notations are specified in Table.~\ref{tab:strategies}.
\begin{table}[!htbp]
    \centering    
    \small{
    \begin{tabular}{c|c|c| c}
    \hline
        NE in $G$  &  ASW strategies in $G$ &  NE in $G^2$ & ASW strategies in $G^2$  \\
        $\langle \pi_1,\pi_2\rangle$ & $\langle \pi_1^\asw,\pi_2^\asw\rangle $ & $\langle  \pi_1^2,\pi_2^2 \rangle $ & 
        $\langle \pi_1^{2,\asw},\pi_2^{2,\asw} \rangle $\\
\bottomrule    \end{tabular}}
  \caption{Notations for players' strategies.}
    \label{tab:strategies}
\end{table}

Further, we restrict P1's deceptive strategy to the following class of strategies.
\begin{definition}[One-time switching strategy]
A switching strategy is a function $\switch: \prefplays \rightarrow \dist{\pi_1\cup \pi_1^2}$ that assigns, for a history $\rho \in \prefplays $, a probability distribution over the two best responses, $\pi_1^2$ for P1 in game $G^2$ and $\pi_1$ for P1 in game $G$. The switching strategy is \emph{one-time} if it satisfies the following condition: For any $\rho \in \plays$, there exists a \emph{switching point} $t$ such that $\switch(\rho[0:k]) = \pi_1^2$ for all $k\le t$ and $\switch(\rho[t+1:t+n])=\pi_1$ for all $n >1$. 
\end{definition}

It is noted that P1 may not use the hidden action \emph{immediately} upon the switching. For example, the best response $\pi_1$ may not employ a hidden action till a later time after the switching time. Still, when P1 deviates from P2's perceived best response $\pi_1^2$ for P1,   it is possible for P2 to detect a mismatch of the observed game play from her perceptual game $G^2$, albeit with some delay. Thus, P1's deceptive planning must incorporate a theory of mind for P2 and a reasonable detection mechanism that P2 can use. Next, we show that P1's theory of mind for P2 is inherently incomplete. 

Assuming that P2 can detect the deviation of P1 at some time $t+k_1$, for $k_1\ge 0$. In P1's theory of mind of P2, the strategy of P2 shall be \emph{\ac{bsr}} \cite{sasakiSubjectiveRationalizabilityHypergames2014} in the hypergame, defined as follows. %

\begin{definition}[P1's Incomplete Model of P2's Behaviorally Subjectively Rationalizable Strategy]
\label{def:bsr}
P1's incomplete model of a behaviorally subjectively rationalizable strategy (BSR) for P2 is a   function $\pi_2^B: \prefplays \rightarrow \dist{A_2}\cup \uparrow$. The function $\pi_2^B$ is constructed as follows: For any history $h =s_0 (a_0,b_0)s_1 (a_1,b_1)s_2\ldots s_n \in \prefplays$ of length $n$, let $h_{0:t}$ be the history during which P1 follows the best response $\pi_1^2$ in $G^2$ and $h_{t+1: n} $ be the history during which P1 follows the best response $\pi_1$ in the true game $G$, let $k_1 $ be the time step when P2 detects the deviation and $k_2$ be the time step when P2 learns about P1's true action set $A_1$, it holds that:
 \begin{itemize}
     \item  For all $0 \le i \le t+k_1, $
  \[
    \pi_2^{B }(h_{0:i}) = \begin{cases}  \pi_2^2(s_i), & \text{ if } s_i \notin \asw_2^2,\\
    \pi_2^{2,\asw}(s_i), & \text{ otherwise.}
    \end{cases}  
  \]
  That is, P1 predicts that P2 follows the best response in the game $G^2$.
\item   For all $  t+k_1 <i \le t+k_1+k_2  $, 
  \[
    \pi_2^{B}(h_{0:i}) = \uparrow,
   \]
   
   That is, P1 cannot predict what strategy P2 will follow during this time span. Thus, the strategy is undefined.
   
 \item   And for all $i > t+k_1+k_2$, 
  \[  \pi_2^{B}(h_{0:i}) =\begin{cases}  \pi_2 (s_i),  & \text{ if } s_i \notin \asw_2,\\
  \pi_2^{\asw} (s_i),  & \text{ otherwise.}
  \end{cases}
  \]
  That is,  P2 follows the best response in the game $G$.
 \end{itemize} 
\end{definition}

P1's model of P2's \ac{bsr} strategy is incomplete because P1 cannot predict what strategy P2 will employ once P2 detects that the game she knows is incorrect but does not yet know what the true game is. In the case when P2 cannot observe P1's actions, it is possible that P2 will not learn the true game dynamics given her partial observations.  In this model, P2 always commits to her subjective rationalizable strategy in his perceptual game, whether it is $G^2$ in the beginning or $G$ after learning P1's actions. There is no advantage to deviate from the subjective rationalizable strategy.

A \emph{completion} of P1's model of P2's \ac{bsr} strategy, denoted by $\tilde \pi_2^B$, is defined such that for any $h \in \prefplays$, if $\pi_2^B(h)$ is defined, then $\tilde \pi_2^B (h) =  \pi_2^B(h)$, otherwise $\tilde \pi_2^B(h)\in \dist{A_2}$ can be an arbitrary distribution over P2's actions. We define $\tilde \Pi_2^B$ be the set of \emph{all possible completions} for $\pi_2^B$. With this notion, we can define the value of action-deception as follows.

\begin{definition}[The value of action deception using one-time switching strategy]
\label{def:VoD}
For any one-time switching strategy $\switch$ of P1, the \emph{value of action deception} for any initial state $s_0 \in S\setminus (\asw_1 \cup \asw_2)$, \ie, the  positive winning region for P1/P2, 
\[
\mathsf{VoD}(\switch)= \min_{\tilde \pi_2^B \in \tilde \Pi_2^B}{u_1(s_0, \switch, \tilde \pi_2^B)} - {u_1(s_0, \pi_1,\pi_2)},
\]  
and the optimal one-time switching strategy $\optswitch$ is such that 
\[
\optswitch = \arg\max_{\switch \in \Pi_1^{\mathsf{sw}}}
\mathsf{VoD}(\switch),
\]
where $\Pi_1^{\mathsf{sw}}$ is the set of one-time switching strategies in which P1 can select. The optimal value of action deception is 
$\mathsf{VoD}(\optswitch)$.
\end{definition}
By definition, if the value of action deception is greater than 0, then P1 will gain more payoff against P2  by using action deception than what P1 should have obtained if P1 informs P2 of the true game dynamics. Because P1 cannot predict how P2 reacts upon detecting the game mismatch, P1's computation of the value of action deception considers the worst case completion of his incomplete model of P2's \ac{bsr} strategy.

\begin{example}
\label{ex1}
 

We use a variation of the soccer game to illustrate the two different key events when P1 employs a one-time switching strategy.
\begin{itemize}
    \item[1.] P1 switches his strategy from $\pi_1^2$ to $\pi_1$.
    \item[2.] P2 detects that P1 deviates from the equilibrium in P2's perceptual game $G^2$.
    \item[3.] P1 uses his hidden action.
\end{itemize}
Consider the arena in Figure \ref{fig:Bouncing}, where the blue cells represent bouncing walls. We assume $A_1^2=A_2 = \{a_U, a_D, a_L, a_R\}$ and a hidden action for P1 is that P1 can traverse the yellow cell. But in P2's perceptual game, that yellow cell is a bouncing wall. The target states for P1 and P2 are the same as those in the soccer game introduced in Subsection~\ref{sec:soccer}. 

Since P2 does not know the hidden action of P1, P2's subjective rationalizable strategy in $G^2$ will inform  P2 to reach the starred cell, where she can intercept P1 with the highest probability, given P1's best response in $G^2$. P2 predicts that P1 will also move to the top row.  In the meantime, an optimal strategy for P1 is to move toward the bottom corridor and eventually use his hidden action to win. In this example, P1 will switch his strategy to $\pi_1$ as the game starts and then use his hidden action when the yellow cell is reached, which shows that event 3 occurs after event 1, and with a possible delay. 
In this example, events 1 and 2 occur at the same time.
\begin{figure}[!h]
  \centering
  \includegraphics[width=0.6\linewidth]{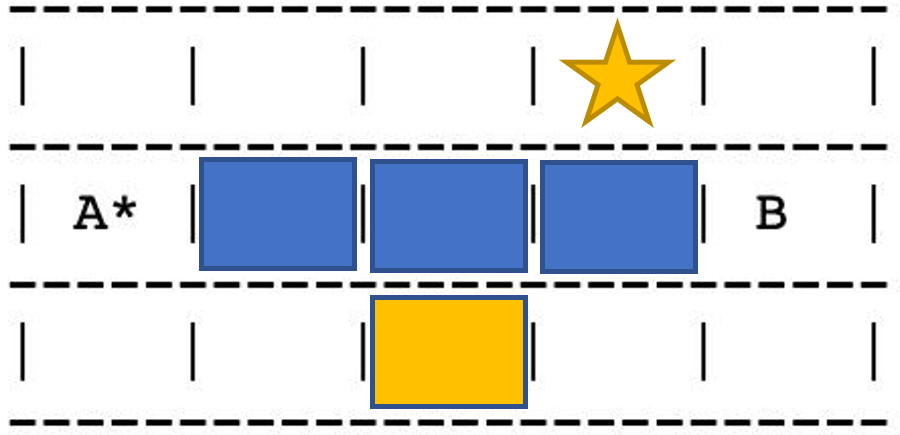}
  \caption{Soccer game with bouncing walls. }
  \label{fig:Bouncing}
  \centering
\end{figure}
In general, 2 only occurs after 1, with an inherent delay introduced by change detection algorithms. 
\end{example}

\subsection{A Detection Mechanism}
To complete a theory of mind for P2, we incorporate a change detection mechanism that predicts the detection time $k_1$.

\begin{assumption}
P2 should not detect any deviation if P1 follows the best response $\pi_1^2$ in P2's perceptual game $G^2$.
\end{assumption} 

Given P2's perceptual game $G^2$, P2's knowledge about her action sequence $b_0 b_1,\ldots, b_n$, and P2's knowledge about P1's best response  $\pi_1^2$ in $G^2$, P2's null hypothesis is a Markov chain $\mc_0=\{S_t, t\ge 0\}$ where $S_0= s_0$ is the  initial game state  and the transition function  is defined by
\begin{multline*}
    \Pr_0(S_{t+1} =s' \mid S_t =s, B_t = b_t) =
    \\
    \sum_{a \in A_1^2}    P^2(s' \mid  s, a, b_t)\pi_1^2(s,a), \text{ for } t=0,1,\ldots, n.
\end{multline*}

Had P2 known that P1 will switch to strategy $\pi_1$, which is P1's best response in the true game, he can construct the alternative hypothesis as another Markov chain  $\mc_1=\{S'_t, t\ge 0\}$ with the initial state $S'_0=s_0$ and the transition function 

\begin{multline*}
\Pr_1(S'_{t+1} =s' \mid S'_t =s, B_t = b_t) =
\\
\sum_{a \in A_1 }    P^2(s' \mid  s, a, b_t)\pi_1 (s,a), \text{ for } t=0,1,\ldots, n.
\end{multline*}

Assuming P2 knows both $\mc_0$ and $\mc_1$, then P2 can employ Page's likelihood ratio CUSUM change point detection \cite{lai1995sequential} to detect if P1 switched the strategy.  Given the observation  $\{s_0, \dots, s_N\}$ generated by  $\mc_0$ for the first $k$ steps and then $\mc_1$ for the remaining $k+1,\ldots, N$ time steps, with the knowledge of P2's action sequence $b_k, k=0,\ldots, N-1$, the change detector is to determine the change point.
The stopping time $N_G$ is defined as follows.
\[
    N_G = \inf \{n \mid \max_{1 \leq k \leq n} \Big[ \sum_{i=k}^n \log\frac{\Pr_1(s_{i}\mid s_{i - 1}, b_{i-1})}{\Pr_0(s_{i}\mid s_{i - 1}, b_{i-1})} \Big] \geq c_\gamma\},
\]
where $c_\gamma$ is a user-defined threshold and $1 \leq n \leq N_G$.   We introduce the following discrimination function: Let $s_{0:n}= s_0, b_0, s_1, b_1,\ldots, s_n$ be P2's observation up to time $n$,
\[
    d(s_{0:n}) = \max_{1 \leq k \leq n} \Big[ \sum_{i=k}^n \log\frac{\Pr_1(s_{i}\mid s_{i - 1}, b_{i-1})}{\Pr_0(s_{i}\mid s_{i - 1}, b_{i-1})} \Big].
\]
The update of the discrimination function can be made incremental as follows,
\begin{multline*}
d( d(s_{0:n-1}), (s_{n-1}, b_{n-1}, s_n)) =\\ \max \Big\{ d(s_{0:n-1}) + \log \frac{\Pr_1(s_{n}\mid s_{n - 1}, b_{n-1})}{\Pr_0(s_{n}\mid s_{n - 1}, b_{n-1})}, 0 \Big\}.
\end{multline*}

\begin{remark} If P2 can observe P1's action sequence $a_0a_1,\ldots ,a_n$, then we only need to construct the Markov chain over the set $S\times A_1$.
\begin{multline*}
    \Pr_0(S_{t+1} =s', A_t= a_t \mid S_t =s, B_t = b_t) =
    \\
      P^2(s' \mid  s, a_t, b_t)\pi_1^2(s,a_t), \text{ for } t=0,1,\ldots, n.
\end{multline*} 
The chain for the alternative hypothesis is constructed analogously. 
\end{remark}

The assumption that P2 knows the alternative hypothesis $\mc_1$ is unrealistic. We refer to this P2 with such knowledge as an \emph{informed} opponent. 


 
Next, we show that P1's deceptive planning strategy against such an informed opponent will provide a lower bound on the performance for P1's deceptive planning against P2 in the actual game. Our formulation employs a semi-\ac{mdp}, which is a class of \ac{mdp} in which the agent selects policies rather than primitive actions. The semi-\ac{mdp} is equivalently expressed as a one-player stochastic game where P1 makes a decision, and then the nature player determines a stochastic outcome. In this way, we can capture clearly how nature's choice affects the theory of mind of P2, constructed by P1.
 

\begin{definition}[Planning an optimal one-time switch action deception]
\label{Meta-planning}
The planning with action deception can be formulated as a semi-Markov decision process
\[
M = (V_1 \cup V_N ,  \{\pi_1^2,\pi_1\}, \Delta, R),
\]
where 
\begin{itemize}
    \item $V_1 \coloneqq  S\times \Phi \times \{0,1\}$ is the state set at which P1 makes a decision.
    Each state $(s, \phi, \bool)$ includes a state $s\in S$ of the original game $G$, a real number $\phi \in [0, + \infty)$, and a Boolean $\bool \in \{0,1\}$. Here $\phi$ represents the value of the discrimination function given some history. The Boolean $\bool$ keeps track of whether P1 has deviated from the equilibrium in $G^2$.

    \item 
    $V_N\coloneqq S\times  \Phi \times \{0,1\} \times A_1\times A_2$ is the state set at which nature determines a probabilistic outcome.
    \item $\pi_1^2,\pi_1$ are two macro-actions (policies) for P1.
    \item $\Delta$ is the probabilistic transition function, defined for both P1 and the nature player's states.  
    
First, at any nature's state $v = (s,\phi,\bool, a,b)$, if $s \in \asw_1 \cup \asw_2$ (almost-sure winning regions  for either player in the true game $G$), then with probability one, a sink state $\sink$ is reached, \ie, 
\[
\Delta((s,\phi,\bool,a,b),\lambda ,\sink)=1.
\]
where $\lambda$ is a null action, representing nature's probabilistic choice.

Second, consider a P1's state $(s, \phi, \bool)$,  there are three cases:

Case 1: if $\bool=0$ and P1 chooses the strategy $\pi_1^2$, then 
\[
\Delta((s,\phi,0), \pi_1^2, (s, \phi, 0, a, b )) =  \pi_1^2(s,a)\cdot \pi_2^2(s,b)
\]

At the state $(s, \phi, 0, a, b )$, the nature determines a probabilistic outcome
\[
\Delta((s,\phi,0, a, b),\lambda, (s', \phi , 0)) = P(s,(a, b),s').
\]

Case 2: if $\bool=0$ and 
 P1 switches to strategy $\pi_1$, then 
\[
\Delta((s,\phi,0), \pi_1, (s, \phi , 1, a, b )) =  \pi_1(s,a)\cdot \pi_2^2(s,b),
\]
where the Boolean switches from $0$ to $1$ indicating that P1 switched strategies. Then, 
at the state $(s, \phi, 1 , a, b )$, the nature determines a probabilistic outcome, \[
\Delta((s,\phi,1, a, b ), \lambda, (s', \phi', 1)) = P(s,(a,b),s'),
\]
where $\phi' = d(\phi, \obs(s,(a,b),s'))$ is the updated value for the discrimination function given P2's observation of the transition.

Case 3: If $\bool=1$, P1 only has one macro-action, which is to follow his best response in the game $G$. Consider a state $(s, \phi, 1)$, then 
\[
\Delta((s,\phi,1), \pi_1, (s, \phi, 1, (a,b))) =  \pi_1(s,a)\cdot \pi_2^2(s,b).
\]
Note that P2 still follows the NE in her perceptual game $G^2$.

Given the state $(s, \phi, 1, (a,b))$, the nature decides the next state probabilistically. If $\phi \leq c_\gamma$, then
\[
\Delta((s,\phi,1,(a,b)),\lambda, (s', \phi', 1)) = 
P(s,(a,b),s').
\]
where  $\phi' = d(\phi, \obs(s,(a,b),s'))$.

Otherwise, if $\phi  > c_\gamma$,
\[
\Delta((s,\phi,1,(a,b)), \lambda, \sink) = 1.
 \]



\item ${R}\colon V_1\cup V_N \to \mathbb{R}$ is a state-based reward function. For any $v\in V_1$, $R(v)=0$, and for any $v =(s, \phi, \bool,a,b)\in V_N$,
\begin{multline*}
    {R}(v) = 
    \begin{cases}
    1,  \quad \text{if} \;\; s \in \asw_1,\\
    -1, \quad  \text{if}\;\;  s \in \asw_2, \\
    u_1(s, \pi_1,\pi_2), \quad \text{if} \;\; \phi > c_\gamma \land s \not\in \asw_1 \cup \asw_2,\\
    0, \quad \text{otherwise. }
    \end{cases}
\end{multline*}
\end{itemize}

\end{definition}

In this planning problem, the process terminates when one of the players reaches his/her almost-sure winning regions in the true game, or P2 detects the deviation of P1's best response from his perceptual game. The semi-\ac{mdp} can be solved using dynamic programming algorithms.

\begin{lemma}
\label{lma:subsetasw}
Given that $G^2$ can be obtained from $G$ by eliminating all transitions enabled by pairs of  P1's hidden actions $A_1\setminus A_1^2$ and P2's actions, it holds that $\asw_2 \subseteq \asw_2^2$.
\end{lemma}

The proof follows from the computation of almost-sure winning regions \cite{deAlfaro2000} in a concurrent reachability game and thus is omitted.  Intuitively, P1's hidden actions can make a state $s\in \asw_2^2\setminus \asw_2$ becomes positive winning for P1 and P2, regardless of P2's action choice.

The following statement holds:
\begin{theorem}
Let $\optsemi$ be the optimal policy in the semi-\ac{mdp} $M$,   it holds that \[
 \mathsf{VoD}(\optsemi)   \le \mathsf{VoD}(\optswitch). 
\]
where $\optswitch$ is the  optimal one-time switching strategy for action deception (Def.~\ref{def:VoD}).
\end{theorem}
\begin{proof}For each   state sequence that terminates in $\sink$ in the semi-\ac{mdp}, 
$
\tilde \rho = (s_0, \phi_0, \bool_0)  (s_0, \phi_0, \bool_0, a_0,b_0)   (s_1, \phi_1,\bool_1),   \ldots\\ (s_N, \phi_N, \bool_N, a_N, b_N),   \sink$, 
we can identify one unique play in the original game  $\rho = s_0 (a_0,b_0)s_1 (a_1,b_1)s_2\ldots s_N \in \prefplays. $

Following the previous analysis of P2's \ac{bsr} strategy, let $t$ be P1's strategy switching time, $t+k_1$ be the time when P2 detects a deviation and $t+k_1+k_2$ be the time when P2 learns about the true game.
The following cases are possible:

Case 1: If $t+ k_1 > N$, which means that P2 does not detect P1's deviation from $\pi_1^2$ and thus for all $0\le k \le N$, $h_k \le c_\gamma$, then the game can only terminate if $v_N \in \asw_1\cup \asw_2$. In the case that $s_N\in \asw_1$, P1 receives a payoff of 1 as he can use the almost-sure winning strategy starting from $s_N$. In the second case, $s_N\in \asw_2$, P1 receives a payoff of -1. Because $\asw_2\subseteq \asw_2^2$ (Lemma~\ref{lma:subsetasw}) and P2 does not know the true game $G$, P2 can commit to the almost-sure winning strategy $\pi_2^{2,\asw}$. Such a strategy is sub-optimal in the true game. Thus, P1's payoff of $-1$ is a lower bound on the actual payoff P1 may receive because P2's suboptimal strategy may provide a possibility for the game states to leave $\asw_2$. This is because a policy almost-sure winning for $\asw_2^2$ may not be almost-sure winning for $\asw_2$ based on the qualitative analysis of reachability objectives \cite{deAlfaro2000}. 

Case 2: If $t+ k_1=N$, which means that P2 detects P1's deviation from $\pi_1^2$ at time step $N$, then P1's reward is given by the payoff of the equilibrium in the true game. This reward is a lower bound because it assumes that after detection, P2  commits to his best response in $G$. This assumption ignores the possible delay that P2 learns about the true game after detection. For any strategy $\hat \pi_2$ that P2 can commit to after the detection, we have that  $ u_2(s_N,\pi_1, \hat \pi_2) \le  u_2(s_N,\pi_1, \pi_2) $ and thus $ u_1(s_N,\pi_1, \hat \pi_2) \ge  u_1(s_N,\pi_1, \pi_2) $ because $\langle \pi_1,\pi_2\rangle$
is the equilibrium in the zero-sum game $G$. 

Given both cases, the reward P1 obtains upon reaching $\sink$, is a lower bound on the actual payoff P1 receives against a \ac{bsr} strategy employed by P2. In addition, if the process does not terminate, then the reward obtained by P1 is zero, which is the same as the reward of a non-terminating play for P1 against any \ac{bsr} strategy of P2. 

Thus, the optimal value of action deception is lower bounded by the value of the optimal policy in the semi-\ac{mdp} $M$.
  \end{proof}


\section{Experiments}

We illustrate the solution using the soccer game example in Example~\ref{ex1}. First, it is observed that the \ac{mdp} in Def.~\ref{Meta-planning} has hybrid state space because the range $\Phi$ of the discrimination function is continuous. We employ a discretization-based approach to solve the \ac{mdp} by uniformly discretizing state-space $[0,c_\gamma]$ into $n$ intervals, $[(i - 1)\delta, i \delta], i = 1, 2, \dots, n$, where $\delta$ is the length of the interval. For a discrimination function value $\phi \in [(i - 1)\delta, i \delta]$, we label it as a discrimination function value level $\phi_i$.  For $\phi > c_\gamma$, we label it as $\phi_{ex}$. In the update of the discrimination function value, the midpoint of the interval represents the level $\phi_i$. For instance, if the current level is $\phi_i$, the update will be $\phi' = d(\frac{(2i - 1)\delta}{2}, \obs(s,(a,b),s'))$.
In the experiments, we set $\delta = 0.2$, $c_\gamma = 2$. Hence, there are $10$ levels with $[0, c_\gamma]$ and a level $\phi_{ex}$. 


\subsection{Value of Deception and Comparative Analysis}
We show the value of action deception by comparing the differences between some state values under different strategies. Since there are too many states in the MDP. We mainly focus on initial states, i.e., the state that $s \not\in \asw_1 \cup \asw_2$, $\phi = 0$, and $\bool = 0$.


Figure~\ref{fig:True_1} uses heat maps to show the value of action deception $\mathsf{VoD}(\optsemi)$ given the optimal switching strategy $\optsemi$ obtained by solving the semi-MDP. To make the results clearer,  we plot the figure by multiplying the value by 100, \ie,  the range of $\mathsf{VoD}(\optsemi)$ is enlarged to $[0,100]$ instead of $[0,1]$. We employ value iteration to solve the semi-\ac{mdp} that terminates when the Bellman error is below $0.1$ (with respect to the enlarged reward).
\begin{figure}[!ht]
  \includegraphics[width=\linewidth]{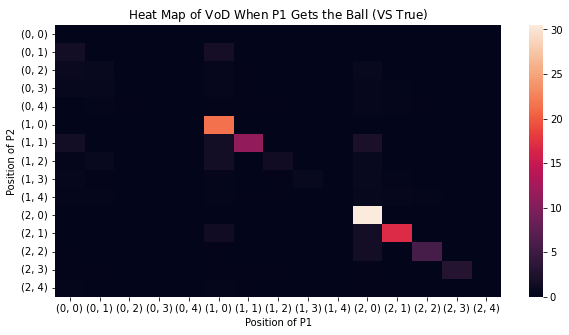}
  \includegraphics[width=\linewidth]{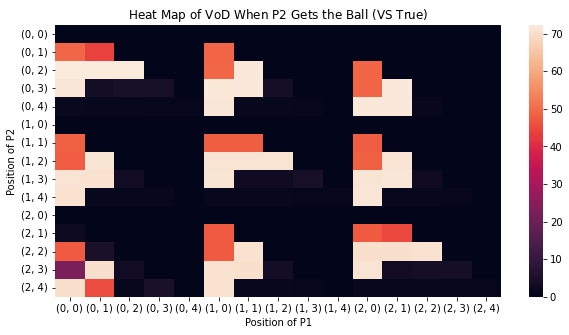}

  \caption{The $\mathsf{VoD}(\optsemi)$ given all possible initial positions.}
  \label{fig:True_1}
\end{figure}
From this figure, we observe that P1 will benefit more from action deception when P2 has the ball at the beginning of their interaction.   Across all initial states, the maximal VoD is $72.289$ and the minimal VoD is $0$. In other words, the maximum gain of winning probability for P1 is nearly $70\%$ for some states. If from a given initial state where $\mathsf{VoD}(\optsemi)$ is close to zero, then P1 may choose not to deviate and instead inform P2 of the true game.  

To gain more insight into P1's switching policy, Figure~\ref{fig:Switch_State} shows a snapshot of the game state where P1 switches his strategy.
\begin{figure}[!ht]
  \includegraphics[width=0.45\linewidth]{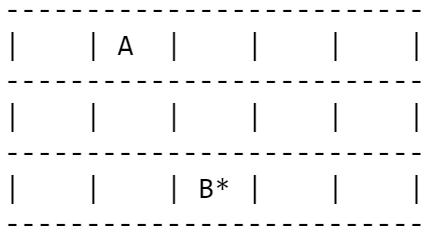}  
  \caption{P1 switches his strategy to avoid losing the game.}
  \label{fig:Switch_State}
\end{figure}
This state $s \in \asw_2^2$, i.e., the almost-sure winning region of P2 in $G^2$. P2 thinks that she will win with probability $1$ by moving left. However, if P2 knows the true game, P2 will not go left now since she knows that P1 can intercept the ball with probability $0.5$ by using the hidden action $a_H$. That is, $\pi_2^2(a_L \mid s) = 1$ and $\pi_2(a_L \mid s) = 0$ at the current state $s$. When concurrently, P2 moves left according to $\pi_2^2$ and P1 uses the hidden action to move down two cells, P1 increase his chance of winning by $50\%$.



To understand how the delay in detection can be exploited by P1, we perform the following experiment. Instead of using a realistic change detection, the semi-\ac{mdp} is constructed assuming that P2 has no delay in detecting the change. Then, by solving this semi-\ac{mdp}, we obtained an optimal switching strategy $\so$ for P1 to play against such a strong opponent P2. 

Figure~\ref{fig:VS_SO} uses heat maps to show the differences of state value, i.e, $u_1(s_0, \optsemi, \pi_2^2) - u_1(s_0, \so, \pi_2^2)$ for different initial states. Note that the difference equals the difference between the value of deception against the realistic P2 and the value of deception against the strong opponent.
\begin{figure}[!ht]
  \includegraphics[width=\linewidth]{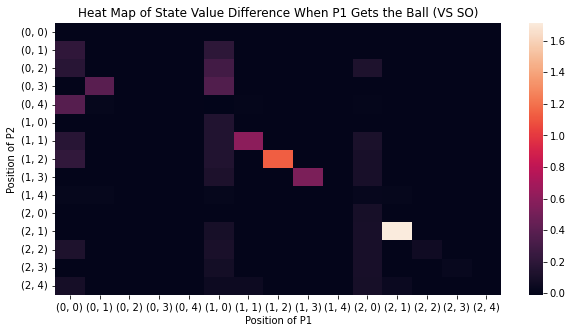}
  \includegraphics[width=\linewidth]{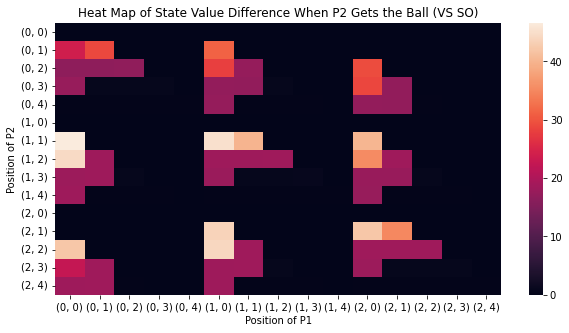}
  \caption{The difference between $\mathsf{VoD}$  for the realistic P2  and a  strong opponent P2 with no detection delay.}
  \label{fig:VS_SO}
\end{figure}
Across all initial states, the maximal state value difference is $46.564$ and the minimal state value difference is $0$.  We also observed that the sets of initial states
with higher values of deception are consistent between the case playing against the realistic P2 (Fig.~\ref{fig:True_1}) and the strong opponent P2 (Fig.~\ref{fig:VS_SO}). This result also highlights the importance of exploiting P2's detection delay.



\subsection{Sensitivity Analysis of Detection Threshold}
In the construction of the semi-\ac{mdp}, we fixed the threshold of the change detection algorithm. However, it is possible that the true detection threshold used by P2 can be different. We conduct experiments to assess how effective the deceptive strategies are against uncertainty in the detection threshold.

%
The higher $c_\gamma$ is, the less sensitive the detector becomes. Our previous experiment sets $c_\gamma = 2$ which is a relatively sensitive value. 
Next, we consider different values of  $c_\gamma = 1, 5, 8, 12$, respectively. 
For each $c_\gamma$ value, we construct the corresponding semi-\ac{mdp} and then evaluate the strategy computed in the semi-\ac{mdp} when $c_\gamma = 2$ in these different semi-\ac{mdp}s, referred to as $M^1,M^5,M^8, M^{12}$. In this way, we can evaluate how robust P1's strategy performs if P2 employs a detection threshold $c_\gamma=1,5,8,12$ while P1 thinks P2's detection threshold is $c_\gamma=2$. 

The value of $\optsemi$ in $M^i$ is denoted  $u_1^{i}$ for $i=1,5,8,12$, respectively. And the value of $\optsemi$ in the original semi-\ac{mdp} $M$ given $c_\gamma=2$ is denoted $u_1^2$.
The following table (Table~\ref{tab:sensitivity}) shows the maximum values of $u_1^2 (s_0, \optsemi)- u_1^{i}(s_0, \optsemi)$, for each $ i = 1,5,8,12$. From this result, we observe that the performance does not degrade much. At the initial state where the maximum value of $u_1^2 (s_0, \optsemi)- u_1^{i}(s_0, \optsemi)$ is observed, the computed policy $\optsemi$ has a performance degradation within the range of $[2\%,5\%]$.

\begin{table}[H] 
\begin{tabular}{|c|c|c|c|c|}
\hline
              & $c_\gamma=1$ &$c_\gamma=5$ & $c_\gamma=8$&$c_\gamma=12$\\ \hline
maximum difference & 2.235      & 3.529      & 3.464      & 3.403       \\ \hline
original value $u_1^2 (s_0, \optsemi)$ & 92.61   & 94.00  & 94.00     & 94.00 \\ \hline
\end{tabular}
\caption{Comparison of the state value under $\optswitch$ in different semi-\ac{mdp} with $c_\gamma=1,5,8,12$.}
\label{tab:sensitivity}
\end{table}

\section{Conclusion}
In this paper, we develop a planning algorithm for action deception in two-player concurrent stochastic games with asymmetric information in both players' knowledge and observations.  We formally prove that the synthesized switching strategy provides a lower bound on the value of action deception, despite the incomplete information regarding P2’s response strategy.
Building on this result, there are several future directions to be considered: First, whether it is possible to extend the solution concepts from competitive interactions to more general non-cooperative interactions. In practice, asymmetric information is prevalent in multi-agent interactions and it is possible that the agents' intentions can be partially aligned. If one player knows that the other player may not know his action capabilities but is adaptive, how can this player strategically use the private actions to improve multi-agent collaboration? Another direction is to consider action deception in a competitive setting but with a more general information structure, for instance, what if both P1 and P2 have partial observations over state and action sequences? It is interesting to know which subclass of such games may have tractable solutions.

\balance
\begin{acks}
This research was sponsored by the Army Research Office (ARO) and was accomplished under Grant
Number W911NF-22-1-0034 and Grant Number 
W911NF-22-1-0166. 
The views and conclusions contained in this document are those of the authors and
should not be interpreted as representing the official policies, either expressed or implied, of the Army Research
Office or the U.S. Government. The U.S. Government is authorized to reproduce and distribute reprints for
Government purposes notwithstanding any copyright notation herein.
\end{acks}


\bibliographystyle{ACM-Reference-Format} 
\bibliography{refs,jie_refs}


\end{document}

%% file: defs.tex
 \newif\ifuseboldmathops
\newif\ifuseittextabbrevs
\useboldmathopstrue   

\ifuseittextabbrevs

	\newcommand{\ie}{{\it i.e.}}

\else

	\newcommand{\ie}{i.e.~}

\fi

\ifuseboldmathops
	\newcommand{\reals}{\mathbf{R}}

\else
	\newcommand{\reals}{\mathbb{R}}

\fi

\newcommand{\Eventually}{\Diamond \, }

\newcommand{\bool}{\mathsf{b}}
\newcommand{\asw}{\mathsf{ASW}}

\ifuseboldmathops

\else

\fi

\ifuseboldmathops
	\newcommand{\Expect}{\mathop{\bf E{}}\nolimits}
	
\else
	\newcommand{\Expect}{\mathop{\mathbb{E}{}}\nolimits}
	
\fi

\newcommand{\obs}{\mathsf{Obs}}

\newcommand{\dist}[1]{\mathcal{D}(#1)}

 \newcommand{\sink}{\mathsf{sink}}
 
 \newcommand{\plays}{\mathsf{Plays}}
\newcommand{\prefplays}{\mathsf{PrefPlays}}
\newcommand{\ball}{\mathsf{Ball}}

\newcommand{\switch}{\pi^{\mathsf{sw}}}
\newcommand{\so}{\pi^{\mathsf{so}}}
\newcommand{\optswitch}{\pi^{\mathsf{sw},\ast}}
\newcommand{\optsemi}{\pi^{\mathsf{sw},\dagger}}

\newcommand{\mc}{\mathsf{MC}}
\renewcommand{\Pr}{\mathsf{Pr}}

\acrodef{mdp}[MDP]{Markov decision process}

\acrodef{asw}[ASW]{Almost-Sure Winning}